\documentclass[11pt]{article}
\usepackage{amsmath,amsthm,amssymb,epsfig,sectsty,caption,color,times}
\usepackage[letterpaper,hmargin=0.95in,vmargin=1.1in]{geometry}
\usepackage[dvipdfm,bookmarks=false,linkbordercolor={0 1 1},citebordercolor={0 1 1}]{hyperref}

\usepackage{algorithm, algpseudocode}
\algrenewcommand\algorithmicrequire{\textbf{\quad Input:}}
\algrenewcommand\algorithmicensure{\textbf{\quad Output:}}

\theoremstyle{plain}
\newtheorem{theorem}{Theorem}[section]
\newtheorem{lemma}[theorem]{Lemma}

\theoremstyle{definition}

\newcommand{\qedsymb}{\hfill{\rule{2mm}{2mm}}}
\renewenvironment{proof}{\begin{trivlist} \item[\hspace{\labelsep}{\bf \noindent Proof.\/}] }{\qedsymb\end{trivlist}}%
\newcommand{\bs}[1]{\boldsymbol{#1}}

\newcommand{\ceil}[1]{\lceil #1 \rceil}%

\begin{document}

\begin{titlepage}

\title{Improved Approximation for Orienting Mixed Graphs}
\author{
    Iftah Gamzu\thanks{Tel Aviv University.
    Email: \href{mailto:iftah.gamzu@cs.tau.ac.il}{\tt iftah.gamzu@cs.tau.ac.il}.}
    \and
    Moti Medina\thanks{Tel Aviv University.
    Email: \href{mailto:medinamo@eng.tau.ac.il}{\tt medinamo@eng.tau.ac.il}.}
}
\date{}
\maketitle

\begin{abstract}
An instance of the maximum mixed graph orientation problem
consists of a mixed graph and a collection of source-target vertex
pairs. The objective is to orient the undirected edges of the
graph so as to maximize the number of pairs that admit a directed
source-target path. This problem has recently arisen in the study
of biological networks, and it also has applications in
communication networks.

In this paper, we identify an interesting local-to-global
orientation property. This property enables us to modify the best
known algorithms for maximum mixed graph orientation and some of
its special structured instances, due to Elberfeld et al.~(CPM
'11), and obtain improved approximation ratios. We further proceed
by developing an algorithm that achieves an even better
approximation guarantee for the general setting of the problem.
Finally, we study several well-motivated variants of this
orientation problem.
\end{abstract}


\thispagestyle{empty}
\end{titlepage}

\section{Introduction} \label{sec:intro}
An instance of the \emph{maximum mixed graph orientation} problem
consists of a \emph{mixed} graph $G = (V, E_\mathrm{D} \cup
E_\mathrm{U})$ with $n$ vertices, such that $E_\mathrm{D}$ and
$E_\mathrm{U}$ indicate the sets of directed and undirected edges,
respectively. An additional ingredient of the input is a
collection $P \subseteq V \times V$ of source-target vertex pairs.
A source-target vertex pair $(s,t) \in P$ is called a
\emph{request}. The objective is to orient $G$ in a way that
maximizes the number of satisfied requests. An \emph{orientation}
of $G$ is a directed graph $\vec{G} = (V, E_\mathrm{D} \cup
\vec{E_\mathrm{U}})$, where $\vec{E_\mathrm{U}}$ is a set of
directed edges obtained by choosing a single direction for each
undirected edge in $E_\mathrm{U}$. A request $(s,t)$ is said to be
\emph{satisfied} under an orientation $\vec{G}$ if there is a
directed path from $s$ to $t$ in $\vec{G}$.

One may assume without loss of generality that the mixed graph $G$
is \emph{acyclic}, that is, a graph that has no cycles. This
assumption holds since any instance of maximum mixed graph
orientation can be reduced to another instance in which the
underlying mixed graph is acyclic without affecting the number of
requests that can be
satisfied~\cite{SilverbushES11,ElberfeldSDSS11}. Indeed, if the
input graph contains cycles, one can sequentially contract them
one after the other. In each step, the undirected edges of an
arbitrary cycle are all oriented in the same direction. In
particular, if this cycle contains directed edges then the
undirected edges are oriented in a consistent way with those
edges. As a result, every pair of vertices on this cycle admits a
directed path between them, and thus, the cycle can be contracted.
One can easily validate that the resulting mixed acyclic graph
consists of undirected components, each of which must be an
undirected tree, and those components are connected by directed
edges in a way that does not produce cycles. The maximum mixed
graph orientation problem draws its interest from applications in
network biology and communication networks:

\medskip \noindent {\bf Network biology.}
Recent technological advances, such as yeast two-hybrid
assays~\cite{Fields05} and protein co-immunoprecipitation
screens~\cite{Gavin02}, enable detecting physical interactions in
the cell, leading to protein-protein interaction (PPI) networks.
One major caveat of those PPI measurements is that they do not
reveal information about the directionality of the interactions,
namely, the directions in which the signal flows. Since PPI
networks serve as the skeletons of signal transduction in the
cell, inferring the hidden directionality information may provide
insights to the inner working of the cell. Such an information may
be inferred from causal relations in those
networks~\cite{YeangIJ04}. One such source of causal relations is
perturbation experiments, in which a gene is perturbed (cause) and
as a result, other genes change their expression levels (effects).
A change of expression of a gene suggests that the corresponding
proteins admit a path in the network, and in particular, it is
assumed that there must be a directed path from the causal gene to
the affected gene.

Up until this point in time, the above-mentioned scenario can be
modeled as a special instance of the maximum mixed graph
orientation problem in which one is interested to orient the edges
of an \emph{undirected} network in a way that maximizes the number
of cause-effect pairs that admit a directed path from the causal
gene to the affected gene. However, in the more accurate
biological variant, there are several interactions whose
directionality is known in advance. For instance, protein-DNA
interactions are naturally directed from a transcription factor to
its regulated genes, and some PPIs, like kinase-substrate
interactions, are known to transmit signals in a directional
fashion. Therefore, in general, the input network is a mixed
graph.

\medskip \noindent {\bf Communication networks.}
A unidirectional communication network consists of communication
links that allow data to travel only in one direction. One main
benefit of such communication links is that the data of the device
on one side is kept confidential while it may still access the
data of the device on the other side. As a consequence,
unidirectional networks are most commonly found in high security
environments, where a connection may be made between devices with
differing security classifications. For example, unidirectional
communication links can be used to facilitate access to a
vulnerable domain such as the Internet to devices storing
sensitive data. The maximum mixed graph orientation problem
captures the interesting scenario in which one is interested to
design a unidirectional network that maximizes the number of
connection requests that can be satisfied in a secure way. We
remark that unidirectional networks have also been studied in
distributed and wireless ad hoc settings (see,
e.g.,~\cite{AfekG94,AfekB98,MarinaD02} and the references
therein), where a common focus is on algorithmic questions that
arise in a given unidirectional network. Here, we are rather
interested in the question of how to design such a network while
optimizing some performance guarantees.

\subsection{Previous work}
Arkin and Hassin~\cite{ArkinH02} seem to have been the first to
study the problem of orienting mixed graphs. They focused on the
decision problem corresponding to maximum mixed graph orientation,
and demonstrated that it is NP-complete. Elberfeld et
al.~\cite{ElberfeldSDSS11} observed that the reduction in their
proof implies that the maximum mixed graph orientation problem is
NP-hard to approximate to within a factor of $7/8$. Silverbush,
Elberfeld, and Sharan~\cite{SilverbushES11} devised a
polynomial-size integer linear program formulation for this
problem, and evaluated its performance experimentally. Recently,
Elberfeld et al.~\cite{ElberfeldSDSS11} developed several
polylogarithmic approximation algorithms for special instances of
the problem in which the underlying graph is tree-like, e.g., when
the graph has bounded treewidth. In addition, they developed a
greedy algorithm for the general setting that achieves $\Omega(1 /
(M^{c} \log n))$-approximation, where $M = \max\{n,|P|\}$ and $c =
1/\sqrt{2} \approx 0.7071$.

Medvedovsky et al.~\cite{MedvedovskyBZS08} initiated the study of
the special setting of maximum graph orientation in which the
underlying graph is undirected, that is, when there are no
pre-directed edges. They proved that it is NP-hard to approximate
this problem to within a factor of $12/13$, even when the graph is
a star. They also proposed an exact dynamic-programming algorithm
for the special case of path graphs, and a $\Omega(1 / \log
n)$-approximation algorithm for the general problem. Gamzu, Segev
and Sharan~\cite{GamzuSS10} utilized the framework developed in
\cite{GamzuS10a} to obtain an improved $\Omega(\log\log n / \log
n)$-approximation ratio (see also~\cite{ElberfeldBGMSSZS11}). Very
recently, Dorn et al.~\cite{DornHKNU2011} studied this problem
from a parameterized complexity point of view. They presented
several fixed-parameter tractability results. Further research
focused on other variants of this undirected orientation problem.
For example, Hakimi, Schmeichel, and Young~\cite{HakimiSY97}
studied the special setting in which the set of requests contains
all vertex pairs, and developed an exact polynomial-time
algorithm.

\subsection{Our results}
We identify a useful structural property of requests crossing
through a junction vertex. Informally, this property guarantees
that if a set of requests is locally satisfiable then it can also
be satisfied globally. Using this property, we can slightly modify
the algorithms developed by Elberfeld et
al.~\cite{ElberfeldSDSS11}, and obtain improved approximation
ratios. For example, we eliminate a logarithmic factor from their
polylogarithmic approximation ratio for the case that the
underlying graph has bounded treewidth. These results appear in
Section~\ref{sec:localtoglobal}. Although the local-to-global
property can be used in conjunction with the algorithm of
Elberfeld et al.~\cite{ElberfeldSDSS11} to obtain an improved
approximation guarantee for the general setting, we proceed by
developing an improved $\Omega(1 / (n |P|)^{1/3})$-approximation
algorithm for this problem. Our algorithm is based on a greedy
approach that employs the local-to-global property in a novel way.
The specifics of this algorithm are presented in
Section~\ref{sec:general}. We also study two well-motivated
variants of the orientation problem, and most notably, show
hardness results for them. Further details are provided in
Section~\ref{sec:variants}.

\section{From Local to Global Orientations} \label{sec:localtoglobal}
In this section, we identify a useful structural property of
requests crossing through a junction vertex. Informally, this
property guarantees that if there is an orientation of the
\emph{local} neighborhood of a vertex that locally satisfies a set
of requests then it can be extended to a \emph{global} orientation
of the complete graph which satisfies the same set of requests.
Finding a local orientation that maximizes the local
satisfiability is a relatively easy task, namely, it admits a
constant factor approximation algorithm. As a consequence, we can
slightly modify the algorithms developed by Elberfeld et
al.~\cite{ElberfeldSDSS11} so they utilize this property, and
obtain improved approximation ratios. For example, we eliminate a
logarithmic factor from their polylogarithmic approximation ratio
for the special case that the underlying graph has bounded
treewidth.

We associate each request $(s,t) \in P$ with the shortest path $p$
between $s$ and $t$ in the underlying graph. Note that in case
there are several shortest paths for a request, we associate it
with one of them arbitrarily. We now introduce some notation and
terminology. To better understand the suggested notation, we refer
the reader to the concrete example in Figure~\ref{fig:local}.
\begin{itemize}
\item The \emph{local neighborhood} of a vertex $v$ is the
subgraph $G_v$ that consists of $v$, all edges incident on $v$,
and all vertices adjacent to $v$. Notice that the local
neighborhood graph is a star.

\item Let $P_v$ be the set of shortest paths of requests that
cross $v$, and let $P'_v$ be the corresponding set of \emph{local
paths}, that is, the paths of $P_v$ confined to the local
neighborhood of $v$. More precisely, each (global) path $p \in
P_v$ gives rise to a (local) path $p' \in P'_v$ defined as the
intersection of $p$ with the local neighborhood of $v$.
Furthermore, for each $p' \in P'_v$, we define its local endpoints
$s'$ and $t'$ to be the closest vertices to $s$ and $t$ on $p$
that also appear on $p'$, respectively.

\item The \emph{local graph orientation} problem corresponding to
vertex $v$ is defined with respect to the local neighborhood graph
$G_v$ and the set of local paths $P'_v$. The goal is to orient the
undirected edges of $G_v$ in a way that maximizes the number of
satisfied paths in $P'_v$. A path is said to be \emph{satisfied}
if there is a directed path between its source and target vertices
under the orientation.
\end{itemize}

\begin{figure}[!hbt]
\centerline{ \scalebox{0.50}{ 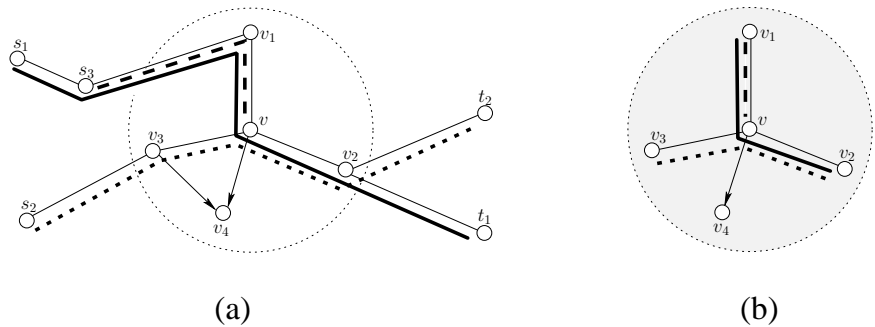 } }
\caption{(a) Suppose $P = \{(s_1,t_1), (s_2,t_2), (s_3 , v)\}$ is
the set of requests, and note that the shortest paths of these
requests are marked with the heavy lines. Notice that all these
paths cross $v$. (b) The local neighborhood of $v$, and the
corresponding set of local paths. For example, notice that the
local endpoints of the request $(s_1,t_1)$ are $s'_1 = v_1$ and
$t'_1 = v_2$.} \label{fig:local}
\end{figure}

\begin{lemma} \label{lemma:localtoglobal}
Given an orientation of $G_v$ that satisfies a set of local paths
$S' \subseteq P'_v$ then there is an orientation of $G$ that
satisfies the corresponding set of global paths $S \subseteq P_v$.
\end{lemma}
\begin{proof}
We argue that if two local paths $p'_1,p'_2 \in S'$ then the
corresponding global paths $p_1,p_2 \in S$ cannot be in conflict.
The paths $p_1$ and $p_2$ are said to be \emph{in conflict} if
they have a mutual undirected edge that gets a different direction
when the edges of $p_1$ are consistently oriented from its source
vertex to its target vertex and when the edges of $p_2$ are
consistently oriented from its source vertex to its target vertex.
Notice that establishing this argument completes that proof of the
lemma since none of the paths of $S$ can be in conflict with
another path in $S$, and therefore, all the paths in $S$ can be
simultaneously satisfied by consistently orienting each one of
them from its source vertex to its target vertex. Note that after
one orients those paths, the remaining undirected edges of the
graph can be oriented in some arbitrary way.

For the purpose of establishing the above argument, let us suppose
that $p_1$ and $p_2$ are in conflict, and attain a contradiction.
Since $p_1$ and $p_2$ are in conflict then there is an undirected
edge $e = (v_1,v_2) \in E_\mathrm{U}$ that gets a different
direction when consistently orienting each one of $p_1$ and $p_2$
from its source vertex to its target vertex. Let us assume without
loss of generality that edge $e$ is the closest to $v$ from all
conflicting edges. We next present a case analysis that depends
whether the edge $e$ appears before or after the position of
vertex $v$ on each of paths $p_1$ and $p_2$. Essentially, there
are two main cases. To better understand the used notation, we
refer the reader to the concrete examples in
Figure~\ref{fig:inconflict}.

\smallskip \noindent {\bf Case I: edge $\bs{e}$ appears after
vertex $\bs{v}$ in both $\bs{p_1}$ and $\bs{p_2}$.} Let us assume
without loss of generality that $v_1$ is closer to $v$ than $v_2$
on $p_1$, and $v_2$ is closer to $v$ than $v_1$ on $p_2$. Let
$d_1$ be the distance between $v$ and $v_1$ on $p_1$, and $d_2$ be
the distance between $v$ and $v_2$ on $p_2$. Since $p_1$ is a
shortest path between $s_1$ and $t_1$, it must also be a shortest
path between $v$ and $v_2$. Thus, $d_1 + 1 \leq d_2$. Similarly,
since $p_2$ is a shortest path between $s_2$ and $t_2$, it must
also be a shortest path between $v$ and $v_1$, and hence, $d_2 + 1
\leq d_1$. Summing together the above inequalities results in $d_1
+ d_2 + 2 \leq d_1 + d_2$, a contradiction.

We note that the case that the edge $e$ appears \emph{before}
vertex $v$ in both $p_1$ and $p_2$ can be handled along the same
lines with an adjustment to the relative position of $v$, e.g.,
the distances need to be defined from $v_1$ and $v_2$ towards the
junction vertex $v$.

\smallskip \noindent {\bf Case II: edge $\bs{e}$ appears after vertex
 $\bs{v}$ in $\bs{p_1}$ and before vertex $\bs{v}$ in $\bs{p_2}$.}
Let us assume without loss of generality that $v_1$ is closer to
$v$ than $v_2$ on both paths $p_1$ and $p_2$. Since $p'_1,p'_2 \in
S'$ we know that the edge on which $p_1$ leaves $v$ and the edge
on which $p_2$ enters $v$ must be different. This implies that the
subpath between $v$ and $v_1$ on $p_1$ and the subpath between
$v_1$ and $v$ on $p_2$ are different. Consequently, merging these
two subpaths creates a cycle in the graph. This contradicts the
fact that the graph is acyclic.

Note that the case that the edge $e$ appears after vertex $v$ in
$p_2$ and before vertex $v$ in $p_1$ is essentially identical to
the above case up to a renaming of the paths.
\begin{figure}[!hbt]
\centerline{ \scalebox{0.65}{ 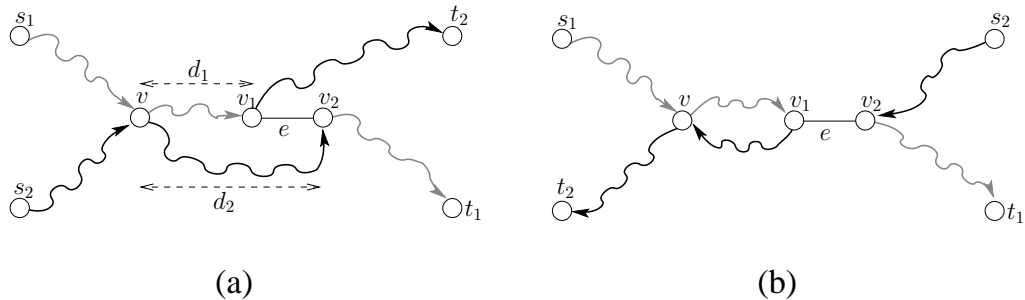 } }
\caption{(a) The case that $e$ appears after $v$ in both $p_1$ and
$p_2$. (b) The case that $e$ appears after $v$ in $p_1$ and before
$v$ in $p_2$.} \label{fig:inconflict}
\end{figure}~
\end{proof}

We now concentrate on the computational complexity of the local
graph orientation problem corresponding to a vertex $v$. One can
easily validate that this problem is equivalent to the maximum
undirected graph orientation problem on a star. Medvedovsky et
al.~\cite{MedvedovskyBZS08} demonstrated that this problem is
equivalent to the maximum directed cut problem. This latter
problem admits constant factor approximation algorithms (see,
e.g.,~\cite{FeigeG95,LewinLZ02}). In fact, one can easily verify
that a random orientation of the undirected edges in the local
neighborhood satisfies at least $1/4$ of the paths of $P'_v$ in
expectation. This follows since the maximal length of any path in
the local neighborhood is at most $2$. Furthermore, one can use
the method of conditional expectations to obtain a deterministic
orientation that satisfies at least $1/4$ of the paths, and
consequently, this approach is a $1/4$-approximation for this
problem. Combining this result with the local-to-global
orientation property exhibited in Lemma~\ref{lemma:localtoglobal}
implies the following theorem.

\begin{theorem} \label{th:localtoglobal}
Given a vertex $v$ and a set of requests $P_v$ whose shortest
paths cross $v$, there is a polynomial-time algorithm that
computes an orientation that satisfies $\Omega(|P_v|)$ requests.
\end{theorem}

We can now modify the algorithms developed by Elberfeld et
al.~\cite{ElberfeldSDSS11} in accordance with
Theorem~\ref{th:localtoglobal}, and obtain the following improved
approximation ratios. We emphasize that the algorithms and their
analysis follow (up to our modification step) those presented by
Elberfeld et al.~\cite{ElberfeldSDSS11}, and thus, we defer them
to the appendix. The first two theorems present algorithms whose
approximation guarantees depend on the treewidth and feedback
vertex number of the underlying graph.

\begin{theorem} \label{th:treewisthalg}
There is a polynomial-time algorithm that finds an orientation
satisfying $\Omega(|P|/(k \log n))$ requests when the undirected
version of the underlying graph has bounded treewidth $k$.
\end{theorem}

\begin{theorem} \label{th:feedbackvertexalg}
There is a polynomial-time algorithm that finds an orientation
satisfying $\Omega(|P|/(k + \log n))$ requests, where $k$ is the
minimum number of vertices whose deletion turns the undirected
version of the underlying graph into a tree.
\end{theorem}

We can also improve the approximation ratios of the algorithms
presented by Elberfeld et al.~\cite{ElberfeldSDSS11} for the
general case, in which there are no structural assumptions on the
graph, by a logarithmic factor.

\begin{theorem} \label{th:generalalg1}
There is a polynomial-time algorithm that approximates the maximum
mixed graph orientation problem to within a factor of $\Omega(1 /
\sqrt{\Delta |P|})$, where $\Delta$ is the maximum length of a
shortest source-target path in the graph.
\end{theorem}

\begin{theorem} \label{th:generalalg2}
There is a polynomial-time algorithm that approximates the maximum
mixed graph orientation problem to within a factor of $\Omega(1 /
M^{1/\sqrt{2}})$, where $M = \max\{n,|P|\}$.
\end{theorem}

Note that we do not provide a proof for the latter theorem since
it can be established along the same lines
of~\cite{ElberfeldSDSS11}, but more importantly, since we next
present an algorithm with a better approximation guarantee.

\section{Improved Approximation for The General Case} \label{sec:general}
In this section, we develop a relatively simple $\Omega(1 / (n
|P|)^{1/3})$-approximation algorithm for the maximum mixed graph
orientation problem. Our algorithm is based on a greedy approach
that employs the local-to-global orientation property developed in
Section~\ref{sec:localtoglobal}.

The algorithm, formally described below, begins by associating
each request $(s_i,t_i) \in P$ with a shortest path $p_i$ between
$s_i$ and $t_i$ in the graph. Then, it greedily orients shortest
paths one after the other until all the remaining paths are in
conflict with many other paths. When this happens, the algorithm
concentrates on the vertex that is crossed by a maximal number of
paths, and utilizes the local-to-global orientation algorithm from
Theorem~\ref{th:localtoglobal} to complete the orientation of the
graph. Recall that two paths $p_1$ and $p_2$ are said to be
\emph{in conflict} if they have a mutual undirected edge that gets
a different direction when the edges of $p_1$ and $p_2$ are
consistently oriented from their source vertex to their target
vertex.

\begin{algorithm}
\caption{Greedy Orientation}\label{cap:Greedy}%
\begin{algorithmic}[1]
\Require A mixed graph $G$ and a collection $P \subseteq V \times V$ of requests %
\Ensure An orientation $\vec{G}$ of $G$ \smallskip %

\State Let $p_i$ be a shortest path for request $(s_i, t_i)\in P$ in $G$, and let $\mathcal{P} = \bigcup\{p_i\}$ %
\While{there is $p_i \in \mathcal{P}$ that is in conflict with less than $(n|P|)^{1/3}$ paths in $\mathcal{P}$} %
    \State Let $\mathcal{Q} \subseteq \mathcal{P}$ be the set of paths in conflict with $p_i$ %
    \State $G \leftarrow$ the graph that results by orienting the edges of $p_i$ from $s_i$ towards $t_i$ in $G$%
    \State $\mathcal{P} \leftarrow \mathcal{P} \setminus (\mathcal{Q} \cup \{p_i\})$ %
\EndWhile %
\State Let $v$ be a vertex that a maximal number of paths in $\mathcal{P}$ cross, and let $\mathcal{P}_v \subseteq \mathcal{P}$ be that set of paths%
\State $G \leftarrow$ the graph that results by executing the algorithm from Theorem~\ref{th:localtoglobal} with respect to $v$ and $\mathcal{P}_v$ \smallskip%
\State \textbf{return} $G$ %
\end{algorithmic}
\end{algorithm}

One can easily verify that the algorithm computes a feasible
orientation, namely, it assigns a single direction to each
undirected edge. This follows since no conflicting paths are
oriented during the main loop of the algorithm, and since the
algorithm from Theorem~\ref{th:localtoglobal} is known to compute
a feasible orientation. We next prove that the algorithm satisfies
$\Omega(1 / (n|P|)^{1/3})$-fraction of all requests. Clearly, this
implies that the algorithm achieves (at least) the same
approximation guarantee.

\begin{theorem}\label{thm:generalalg}
The greedy orientation algorithm satisfies $\Omega(1 /
(n|P|)^{1/3})$-fraction of all requests.
\end{theorem}
\begin{proof}
Let $\mathcal{P} = \bigcup\{p_i\}$ be the initial collection of
shortest paths, and note that $|\mathcal{P}| = |P|$. In addition,
let $\mathcal{P}_2 \subseteq \mathcal{P}$ be the set of paths the
remain after the termination of the main loop of the algorithm,
and $\mathcal{P}_1 = \mathcal{P} \setminus \mathcal{P}_2$.
Finally, let $\mathcal{A}_1$ be the set of paths that our
algorithm satisfies during the main loop of the algorithm, and let
$\mathcal{A}_2$ be the set of paths that the algorithm satisfies
during the execution of the algorithm from
Theorem~\ref{th:localtoglobal}. In what follows, we prove that $|
\mathcal{A}_1| = \Omega(1 / (n|P|)^{1/3}) \cdot |\mathcal{P}_1|$,
and $|\mathcal{A}_2| = \Omega(1 / (n|P|)^{1/3}) \cdot
|\mathcal{P}_2|$. Consequently, we obtain that the number of paths
satisfied by our algorithm is
$$
|\mathcal{A}_1| + |\mathcal{A}_2| =
\Omega\left(\frac{1}{(n|P|)^{1/3}}\right) \cdot (|\mathcal{P}_1| +
|\mathcal{P}_2|) = \Omega\left(\frac{1}{(n|P|)^{1/3}}\right) \cdot
|P| \ .
$$

The fact that $|\mathcal{A}_1| = \Omega(1 / (n|P|)^{1/3}) \cdot
|\mathcal{P}_1|$ easily follows by observing that in each step of
the main loop of the algorithm, one path is satisfied while less
than $(n|P|)^{1/3}$ paths are discarded. Hence, we are left to
prove that $|\mathcal{A}_2| = \Omega(1 / (n|P|)^{1/3}) \cdot
|\mathcal{P}_2|$. We establish a somewhat stronger result by
demonstrating that $|\mathcal{A}_2| = \Omega(1 /
(n|\mathcal{P}_2|)^{1/3}) \cdot |\mathcal{P}_2|$. For this
purpose, consider two paths $p_1,p_2 \in \mathcal{P}_2$ that are
in conflict. We associate the conflict between these paths to an
arbitrary undirected edge that gets a different direction when
$p_1$ and $p_2$ are oriented, and place one token on this edge.
Notice that each path of $\mathcal{P}_2$ is in conflict with at
least $(n|P|)^{1/3}$ other paths in $\mathcal{P}_2$; otherwise,
the main loop would not have terminated. This implies that if we
place a token for each pair of conflicting paths in
$\mathcal{P}_2$ as shown before then the undirected edges of $G$
have at least $(n|P|)^{1/3} \cdot |\mathcal{P}_2| / 2 \geq n^{1/3}
|\mathcal{P}_2|^{4/3} / 2$ tokens placed on them. As a
consequence, there must be a vertex that has at least $t =
|\mathcal{P}_2|^{4/3} / (2n^{2/3})$ tokens placed on the
undirected edges in its local neighborhood. We next argue that if
some vertex has $t$ tokens in its local neighborhood then there
must be $\Omega(\sqrt{t})$ paths that cross that vertex. As a
result, we attain that the number of paths that cross the vertex
$v$, i.e., the vertex that a maximal number of paths from
$\mathcal{P}_2$ cross, is at least $\Omega(\sqrt{t}) =
\Omega(|\mathcal{P}_2|^{2/3} / n^{1/3})$. By
theorem~\ref{th:localtoglobal}, our algorithm satisfies a constant
fraction of these requests, namely, $|\mathcal{A}_2| = \Omega(1 /
(n|\mathcal{P}_2|)^{1/3}) \cdot |\mathcal{P}_2|$, as required.

For the purpose of establishing the above argument, consider some
vertex $u$ that has $t$ tokens in its local neighborhood. Let us
focus on some edge $e$ in this local neighborhood that has $r$
paths that traverse in one direction and $\ell$ paths that
traverse in the other direction. Notice that such an edge is
assigned $r\cdot\ell$ tokens. This implies that if the local
neighborhood of $u$ consists only of the edge $e$ then the minimal
number of paths that cross $u$ corresponds to the solution of
$\min\{r + \ell : r\cdot\ell = t\}$. One can easily verify that
the solution for this expression is $r = \ell = \sqrt{t}$, that
is, the number of paths is $\Omega(\sqrt{t})$. Note that when
there is more than one edge in the local neighborhood of $u$ then
any path may cross at most two edges. As a result, if we denote
the set of edges in the local neighborhood of $u$ by $E_u$, then
the minimal number of paths that cross $u$ dominates the solution
of $\min\{\sum_{e \in E_u} (r_e + \ell_e) / 2 : \sum_e
(r_e\cdot\ell_e) = t\}$; here, $r_e$ and $\ell_e$ indicate the
number of paths traversing edge $e$ in one direction and the other
direction, respectively. One can easily demonstrate that the
solution for the above expression is obtained by assigning
non-zero values only to one pair of $r_e,\ell_e$ variables,
namely, it is equivalent to the solution for the single edge
case.~
\end{proof}

\section{Other Orientation Variants} \label{sec:variants}
In this section, we study two well-motivated variants of the
orientation problem: the first is maximum mixed graph orientation
\emph{with fixed paths}, and the other is maximum mixed
\emph{grid} orientation.

\subsection{Orientation with fixed paths}
We consider the maximum mixed graph orientation \emph{with fixed
paths} problem. This variant is identical to the maximum mixed
graph orientation problem with the exception that each request
$(s,t) \in P$ is also associated with a fixed path $p$ from $s$ to
$t$ in the graph. With this modified definition in mind, a request
$(s,t)$ is satisfied only if the edges of the path $p$ are
oriented from the vertex $s$ towards the vertex $t$. Note that
this variant is seemingly simpler than maximum mixed graph
orientation since the only computational task is to decide which
requests to satisfy, and there is no need to decide which paths
will be used to satisfy those requests. This is also one of our
motivations for studying this variant, hoping that it will shed
some light on the original problem that would lead to a reduction
in the gap between its lower and upper approximation bounds.

We prove that the maximum mixed graph orientation with fixed paths
problem is NP-hard to approximate to within a factor of
$\max\{1/|P|^{1-\epsilon}, 1/m^{1/2 - \epsilon}\}$, for any
$\epsilon > 0$. In fact, we establish this result even when the
underlying graph is undirected. As a consequence, we attain that
this problem is provably harder than the maximum mixed (or
undirected) graph orientation problem, although it may seem
simpler at first glance. Our proof is based on showing that the
problem under consideration captures the well-known \emph{maximum
independent set} problem as a special case.

\medskip \noindent {\bf A hardness of approximation result.}
An input instance for the \emph{maximum independent set} problem
consists of an undirected graph $G' = (V',E')$. The goal is to
find an independent set of maximum size in the graph. An
independent set is a collection of vertices that do not have any
edges between them. This problem is known to be NP-hard to
approximate within a factor of $1/|V'|^{1-\epsilon}$, for any
$\epsilon > 0$~\cite{Zuckerman07}. We next show a value-preserving
reduction from this problem to our maximum mixed graph orientation
with fixed paths problem.

Given an input instance of maximum independent set, we construct
an input instance for our problem that consists of the undirected
graph presented in Figure~\ref{fig:fixedpaths}(a). Specifically,
we begin by creating a graph with $n' = |V'|$ pairs of $s_i,t_i$
vertices, corresponding to the vertices of $G'$, such that each
such pair is connected by a path $p_i$. We intersect all these
connecting paths in a grid-like fashion. Then, each intersection
point is replaced by one of the gadgets exhibited in
Figures~\ref{fig:fixedpaths}(b) and \ref{fig:fixedpaths}(c). The
gadget $g_{i,j}$ that replaces the intersection point of paths
$p_i$ and $p_j$ has $4$ vertices: $v_i$ and $u_i$ that are
appropriately added to path $p_i$, and $v_j$ and $u_j$ that are
appropriately added to path $p_j$. The edges within the gadget has
the form described in Figure~\ref{fig:fixedpaths}(b) if $(i,j)
\notin E'$, or the form described in
Figure~\ref{fig:fixedpaths}(c) if $(i,j) \in E'$. In the latter
case, the path $p_j$ is also modified to consist of the vertices
$v_i$ and $u_i$, so its subpath inside the gadget is $\langle v_j,
u_i, v_i, u_j \rangle$. In addition, the set of requests $P$ for
our problem consist of all $n'$ pairs $(s_i,t_i)$ with their
corresponding path $p_i$.

\begin{figure}[!hbt]
\centerline{ \scalebox{0.61}{ 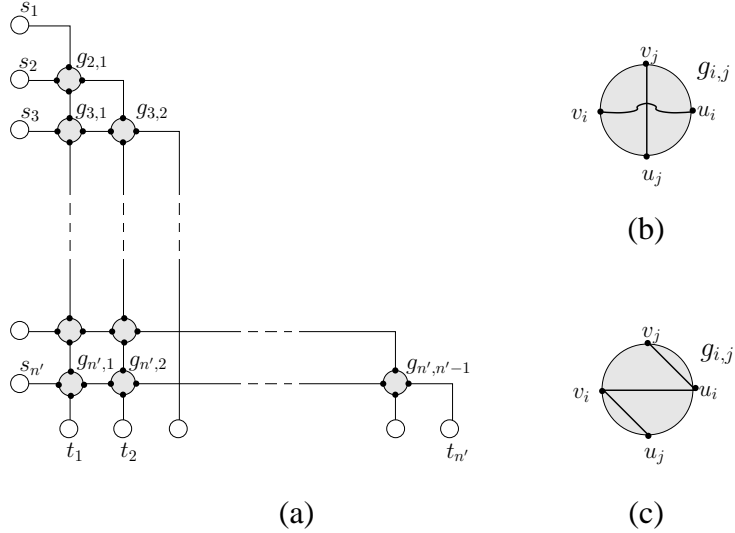 } }
\caption{(a) The graph resulting from the reduction. (b) The
gadget that is used in case $(i,j) \notin E'$. Note that the edges
$(v_i, u_i)$ and $(v_j, u_j)$ do not intersect. (c) The gadget
that is used in case $(i,j) \in E'$.} \label{fig:fixedpaths}
\end{figure}

One can easily validate that a solution $S\subseteq V'$ for the
maximum independent set problem implies an orientation in the
newly-created instance that satisfies the same number of requests.
Specifically, if $i,j \in S$ then clearly $(i,j) \notin E'$, and
thus, the paths $p_i$ and $p_j$ do not share edges. As a result,
one can simultaneously satisfy both request $(s_i,t_i)$ and
$(s_j,t_j)$ by orienting each of their paths from its source
vertex to its target vertex. Conversely, it is not difficult to
verify that given an orientation in newly-created instance that
satisfies some set of requests, one can perform a similar
value-preserving transformation in the opposite direction. In
particular, notice that if $i \notin S$ then there must be $j \in
S$ such that $(i,j) \in E'$, and hence, the paths $p_i$ and $p_j$
require to orient the edge $(v_i,u_i)$ in the gadget $g_{i,j}$ in
conflicting directions. Consequently, one cannot simultaneously
satisfy both underlying requests.

As a result of this value-preserving reduction, and in conjunction
with the hardness result presented by
Zuckerman~\cite{Zuckerman07}, we attain the following
inapproximability result. Recall that $|P| = |V'|$, and notice
that the number of edges in the newly-created instance is $m =
O(|V'|^2)$.

\begin{theorem}
The maximum mixed graph orientation with fixed paths problem is
NP-hard to approximate within a factor of
$\max\{1/|P|^{1-\epsilon}, 1/m^{1/2 - \epsilon}\}$, for any
$\epsilon > 0$.
\end{theorem}

\subsection{Orientation in grid networks}
We study the maximum mixed \emph{grid} orientation problem. This
variant is identical to the maximum mixed graph orientation
problem with the additional restriction that the graph is a grid.
A $n \times m$ \emph{grid} network is a graph with a vertex set $V
= \{1, \ldots, n\} \times \{1,\ldots, m\}$, and an edge set $E$
consisting of horizontal edges, i.e., edges $((i, j),(i, j + 1))$
for all $j = \{1,\ldots,m-1\}$, and vertical edges, i.e., edges
$((i, j),(i+1, j))$ for all $i = \{1,\ldots,n-1\}$. Note that the
study of this variant is motivated by applications in networking.

We prove that the maximum mixed grid orientation problem is at
least as hard as the \emph{maximum directed cut} problem.
Consequently, approximating our problem within factors of $12/13
\approx 0.923$ and $\alpha_{\mathrm{GW}} \approx 0.878$ is NP-hard
and Unique Game-hard, respectively. Interestingly, this finding
comes in contrast with the results attainable for the undirected
grid setting. This latter setting can be solved to optimality in
polynomial-time, and in particular, when the grid is not a path,
that is, when $n,m > 1$, all the requests in $P$ can be satisfied.

\medskip \noindent {\bf A hardness of approximation result.}
An input instance for the \emph{maximum directed cut} problem
consists of a directed graph $G' = (V',E')$. The goal is to find a
directed cut of maximum size in the graph. The size of a cut $A
\subseteq V$ is the number of directed edges $(u,v) \in E'$ such
that $u \in A$ and $v \in V' \setminus A$. Approximating this
problem within factors of $12/13 \approx 0.923$ and
$\alpha_{\mathrm{GW}} \approx 0.878$ is known to be
NP-hard~\cite{Hastad01} and Unique Games-hard~\cite{KhotKMO07},
respectively. In what follows, we present a value-preserving
reduction from this problem to our maximum mixed grid orientation
problem.

Given an input instance of maximum directed cut, we construct an
input instance for our problem which consists of the mixed grid
presented in Figure~\ref{fig:grid}(a). Specifically, we create a
grid whose dimensions are $n = 2|V'|-1$ and $m = 3$. We associate
each vertex $v_i \in V'$ with the vertex $(2i-1, 1)$ in the grid.
The edges incident on each vertex $(2i, 1)$ in the grid are
oriented away from that vertex, and the edges along the perimeter
of the sub-grid that consists of the second and third vertex
columns are oriented in a way that creates a directed cycle. In
addition, the set of requests for our problem is defined to be $P
= E'$.

One can validate that a solution $A \subseteq V'$ for the maximum
directed cut problem implies an orientation in the newly-created
instance that satisfies the same number of requests. Specifically,
if $v_i \in A$ we orient the single undirected edge incident on
vertex $v_i$ of the grid away from that vertex, and if $v_i \in V
\setminus A$ we orient that edge towards vertex $v_i$. Then, it is
easy to see that if an edge of $E'$ is cut by the solution $A$
then the corresponding request is satisfied in the orientation.
Conversely, it is not difficult to verify that given an
orientation in newly-created instance that satisfies some set of
requests, one can perform a similar value-preserving
transformation in the opposite direction. In particular, this side
of the proof builds upon the observation that any request
$(v_i,v_j)$ may only be satisfied by a path that crosses the
undirected edges incident on $v_i$ and $v_j$. The orientation of
all those undirected edges define the cut in the initial problem.

As a result of this value-preserving reduction, and in conjunction
with the hardness results presented by H{\aa}stad~\cite{Hastad01}
and Khot et al.~\cite{KhotKMO07}, we attain the following
inapproximability result.

\begin{theorem}
The maximum mixed grid orientation problem is NP-hard to
approximate within a factor of $12/13 \approx 0.923$, and Unique
Games-hard to approximate within a factor of $\alpha_{\mathrm{GW}}
\approx 0.878$.
\end{theorem}

\begin{figure}[!hbt]
\centerline{ \scalebox{0.62}{ 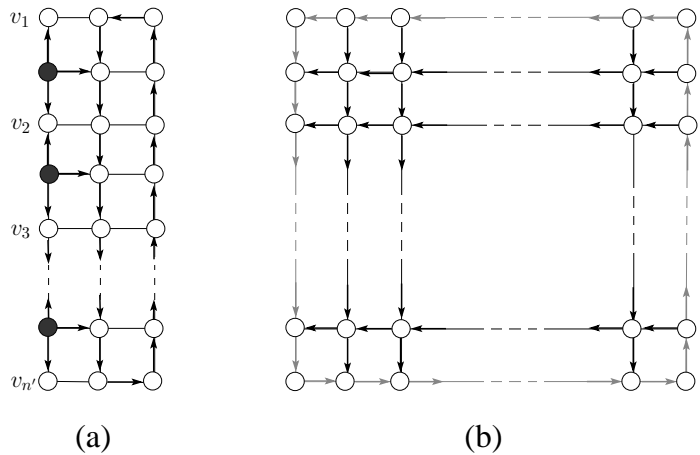 } }
\caption{(a) The grid resulting from the reduction. Note that $n'
= |V'|$. (b) An orientation of an undirected grid that admits a
directed path between any two vertices.} \label{fig:grid}
\end{figure}

\medskip \noindent {\bf Orientation of undirected grids.}
The above-mentioned hardness result comes in contrast with the
results attainable for the undirected grid setting. This latter
setting can be solved to optimality in polynomial-time.
Specifically, when the grid is a path, i.e., when either $m$ or
$n$ equals $1$, there are optimal polynomial-time algorithms for
the problem~\cite{MedvedovskyBZS08,DornHKNU2011}, and when $n,m >
1$, there is a simple orientation that satisfies all the requests
in $P$. This orientation can be obtained by creating a directed
cycle along the perimeter of the grid, and then, orienting all the
remaining horizontal and vertical edges consistently. A concrete
example of such an orientation is presented in
Figure~\ref{fig:grid}(b). One can easily prove that this
orientation admits a directed path between any two vertices of the
graph.

\bibliographystyle{abbrv}
\bibliography{MixedOrientation}

\appendix
\section{Additional Details} \label{appsec:AdditionalDetails}
In this section, we complete the details omitted from the main
part of the paper.

\subsection{Proof of Theorem~\ref{th:treewisthalg}}
The algorithm begins by computing a tree decomposition of width
$k$ for the undirected version of the underlying graph. Note that
this task can be done in polynomial-time since our graph has
bounded treewidth~\cite{Bodlaender96}. A tree decomposition is a
pair $(T, X)$, where $X = \{X_1, \ldots, X_\ell\}$ is a collection
of subsets such that each $X_i \subseteq V$, and $T$ is a tree
whose nodes are the subsets in $X$. Note that $\ell =
\mathrm{poly}(n)$ in our case. The decomposition satisfies the
following properties: (1) $\bigcup X_i = V$, (2) the incident
vertices of every edge of the graph are contained in some subset
in $X$, and (3) if $X_i$ and $X_j$ contain a vertex $v$ then all
the nodes $X_k$ in the unique path between $X_i$ and $X_j$ contain
$v$ as well. The width of the tree is defined to be $\max |X_i| -
1$.

Given the tree decomposition $(T ,X)$, the algorithm proceeds by
computing a shortest path for each request in $P$. The paths are
then classified into at most $\ceil{\log \ell} = O(\log n)$
classes such that for every class, an orientation that satisfies
$\Omega(1/k)$-fraction of its paths can be efficiently computed.
As a consequence, by separately computing an orientation for each
class, and then picking the option that satisfies the highest
number of paths, we are guaranteed to satisfy $\Omega(|P|/(k \log
n))$ of all the requests.

For the purpose of constructing the first class, we find a
centroid node $X_t$ of $T$, that is, a node whose removal breaks
the tree into a collection of subtrees, each of which has at most
half of the vertices in $T$. Note that any tree has a centroid
(see, e.g.,~\cite{FredericksonJ80}). We assign all the paths that
cross a vertex from $X_t = \{v_1,\ldots,v_r\}$ to class
$\mathcal{C}_1$. We further partition $\mathcal{C}_1$ into $r$
collections $\mathcal{C}_{1,1}, \ldots, \mathcal{C}_{1,r}$ such
that a path $p$ is assigned to the collection $\mathcal{C}_{1,j}$
if it crosses $v_j$ but does not cross any of the vertices in
$\{v_1,\ldots, v_{j-1}\}$. Notice that we can satisfy
$\Omega(|\mathcal{C}_{1,j}|)$ paths from the collection $j$ by
applying Theorem~\ref{th:localtoglobal}. One can now easily
validate that executing the mentioned algorithm on each collection
separately, and then picking the option that satisfies the highest
number of paths results in an orientation satisfying
$\Omega(|\mathcal{C}_1|/k)$ requests since $r \leq k+1$.

To construct the second class, we first remove the node $X_t$ from
$T$ to obtain a forest of tree decompositions. For each tree
decomposition, we compute a centroid node, and in the same way as
above, we assign a path to $\mathcal{C}_2$ if it crosses a vertex
from the subsets associated with these centroid nodes. Note that
we only assign paths that were not assigned to the first class.
Using the same arguments as above, we can compute an orientation
that satisfies $\Omega(|\mathcal{C}_2|/k)$ requests. In
particular, one can validate that each path crosses vertices from
exactly one centroid node; otherwise, it should have been assigned
to the first class by properties~(2) and (3) of the tree
decomposition. We now proceed recursively in the same way to
construct the other classes as long as the decompositions under
consideration are not empty. Since the maximal size of a subtree
decreases by at least half in each level of the recursion, this
process terminates within $\ceil{\log \ell}$ steps, and hence,
there are indeed at most $\ceil{\log \ell}$ classes.~

\subsection{Proof of Theorem~\ref{th:feedbackvertexalg}}
The algorithm begins by finding a feedback vertex set $F =
\{v_1,\ldots, v_\ell\}$ in the undirected version of the graph,
namely, a set of vertices whose removal turns the underlying
undirected graph into a tree. Although the computational task of
finding a feedback vertex set with a minimum cardinality is
NP-hard, there is a $2$-approximation algorithm for this
problem~\cite{BeckerG94,BafnaBF99,ChudakGHW98}. Therefore, we may
assume that the cardinality of that set satisfies $\ell \leq 2k$.
The algorithm proceeds by computing a shortest path $p_i$ for each
request $(s_i,t_i) \in P$. Then, each path $p_i$ is classified
into one of $\ell + 1$ classes: if $p_i$ crosses the vertex $v_j$
and none of the vertices in $\{v_1,\ldots, v_{j-1}\}$ then it is
assigned to class $\mathcal{C}_j$; otherwise, if $p_i$ does not
cross any of the vertices of $F$, then it is assigned to class
$\mathcal{C}_{\ell + 1}$. Notice that we can satisfy
$\Omega(|\mathcal{C}_j|)$ paths from any class $j$ by applying
Theorem~\ref{th:localtoglobal}. Also notice that by deleting the
vertices of $F$ from the graph $G$, we obtain a mixed graph which
is a forest of trees, and all the paths in $\mathcal{C}_{\ell+1}$
still remain connected. This mixed tree orientation setting is
known to admit an efficient $\Omega(1/\log n)$-approximation
algorithm~\cite{ElberfeldSDSS11}, and thus, we can satisfy
$\Omega(|\mathcal{C}_{\ell+1}| / \log n)$ paths from the class
$\ell+1$. One can now easily validate that executing the mentioned
algorithms on each class separately, and then picking the option
that satisfies the highest number of paths results in an
orientation satisfying $\Omega(|P|/(k + \log n))$ of all
requests.~

\subsection{Proof of Theorem~\ref{th:generalalg1}}
The algorithm computes a shortest path $p_i$ for each request
$(s_i,t_i) \in P$. Then, it considers those shortest paths in some
arbitrary order, and orients them one after the other. In
particular, when a path $p$ is oriented, all the pending paths
that are in conflict with $p$ are discarded. This greedy
orientation procedure continues as long as a path $p$ under
consideration is not in conflict with more than $\sqrt{\Delta
|P|}$ pending paths. When this happens, there must be some vertex
$v$ on $p$ that at least $\sqrt{\Delta |P|} / \Delta = \sqrt{|P|/
\Delta}$ pending paths cross. This claim holds since the length of
$p$ is known to be at most $\Delta$. The algorithm then employs
the local-to-global orientation algorithm from
Theorem~\ref{th:localtoglobal} with respect to the vertex $v$ and
the corresponding set of pending paths to complete the orientation
of the graph.

One can easily verify that the algorithm computes a feasible
orientation. Therefore, we next prove that the algorithm satisfies
$\Omega(1 / \sqrt{\Delta |P|})$-fraction of all requests. Notice
that if the algorithm employs our local-to-global orientation
algorithm then the resulting orientation indeed attains the
desired performance guarantee since the number of paths satisfied
only by this step is $\Omega(\sqrt{|P|/ \Delta}) = \Omega(1 /
\sqrt{\Delta |P|}) \cdot |P|$. Hence, we may assume that the
algorithm only makes greedy orientation steps. In this case, the
resulting orientation clearly achieves the desired performance
guarantee since one path is satisfied while at most $\sqrt{\Delta
|P|}$ paths are discarded in each of those greedy steps.
\end{document}